\documentclass[runningheads,11pt]{llncs}

\usepackage[margin=1in]{geometry}

\usepackage{xspace}
\usepackage{graphicx}
\usepackage{algorithmic}
\usepackage{algorithm}
\usepackage{amsmath}
\usepackage{amssymb}

\newcommand{\eps}{\varepsilon}
\newcommand{\size}[1]{\left|#1\right|}
\newcommand{\Oh}{\mathcal{O}}

\newcommand{\NP}{\textsf{NP}\xspace}
\newcommand{\FPT}{\textsf{FPT}\xspace}

\newcommand{\kmed}{\textsc{$k$-Median}\xspace}
\newcommand{\capkmed}{\textsc{Capacitated $k$-Median}\xspace}
\newcommand{\colmed}{\textsc{co-$\ell$-Median}\xspace}
\newcommand{\capcolmed}{\textsc{Capacitated co-$\ell$-Median}\xspace}
\newcommand{\support}{\ensuremath{\eps\textsc{-supp}}}

\newcommand{\defparproblem}[4]{
	\vspace{1mm}
	\noindent\fbox{
		\begin{minipage}{0.95\textwidth}
			\begin{tabular*}{\textwidth}{@{\extracolsep{\fill}}lr} \textsc{#1} & {\bf{Parameter:}} #3 \\ \end{tabular*}
			{\bf{Input:}} #2 \\
			{\bf{Task:}} #4
		\end{minipage}
	}
	\vspace{1mm}
}

\begin{document}

\title{To Close Is Easier Than To Open: \\ Dual Parameterization To $k$-Median\thanks{
Part of this work was done while the third and the fifth author were visiting University of Wroclaw. The fifth author was supported by the Foundation for Polish Science (FNP).}}

\titlerunning{To Close Is Easier Than To Open}

\author{Jaros\l{}aw Byrka\inst{1}
\and Szymon Dudycz\inst{1}
\and Pasin Manurangsi\inst{2}
\and \\ Jan Marcinkowski\inst{1}
\and Micha\l{} W\l{}odarczyk\inst{3}}

\authorrunning{J. Byrka et al.}

\institute{University of Wroc\l{}aw, Poland \\ \email{\{jby,szymon.dudycz,jan.marcinkowski\}@cs.uni.wroc.pl} \and
Google Research, Mountain View, USA \\
\email{pasin@google.com} \and
Eindhoven University of Technology, Netherlands \\
\email{m.wlodarczyk@tue.nl}}

\maketitle

\begin{abstract}
The \kmed problem is one of the well-known optimization problems that formalize the task of data clustering.
Here, we are given sets of facilities $F$ and clients $C$, and the goal is to open $k$ facilities from the set $F$, which provides the best division into clusters, that is, the sum of distances from each client to the closest open facility is minimized.
In the \capkmed, the facilities are also assigned capacities specifying how many clients can be served by each facility.

Both problems have been extensively studied from the perspective of approximation algorithms.
Recently, several surprising results have come from the area of parameterized complexity, which provided better approximation factors via algorithms with running times of the form $f(k)\cdot poly(n)$.
In this work, we extend this line of research by studying a~different choice of parameterization.
We consider the parameter $\ell = |F| - k$, that is, the number of facilities that remain closed.
It turns out that such a~parameterization reveals yet another behavior of \kmed.
We observe that the problem is W[1]-hard but it admits a parameterized approximation scheme.
Namely, we present an~algorithm with running time $2^{\Oh(\ell\log(\ell/\eps))}\cdot poly(n)$ that achieves a~$(1+\eps)$-approximation.
On the other hand, we show that under the assumption of Gap Exponential Time Hypothesis, one cannot extend this result to the capacitated version of the problem.
\end{abstract}

\newpage

\section{Introduction}

Recent years have brought many surprising algorithmic results originating from the intersection of the areas of approximation algorithms and parameterized complexity.
It turns out that the combination of techniques from these theories can be very fruitful and a~new research area has emerged, devoted to studying \emph{parameterized approximation algorithms}.
The main goal in this area it to design an algorithm processing an instance $(I,k)$ in time $f(k)\cdot |I|^{\Oh(1)}$,
where $f$ is some computable function, and producing an approximate solution to the optimization problem in question.
Such algorithms, called \emph{FPT approximations}, are particularly interesting in the case of problems for which (1) we fail to make progress on improving the approximation factors in polynomial time, and
(2) there are significant obstacles for obtaining exact parameterized algorithms.  
Some results of this kind are FPT approximations for \textsc{$k$-Cut}~\cite{gupta2018fpt}, \textsc{Directed Odd Cycle Transversal}~\cite{doct}, and \textsc{Planar Steiner Network}~\cite{chitnis2018parameterized}.
A~good introduction to this area can be found in the survey~\cite{feldmann2020survey}.

One problem that has recently enjoyed a significant progress in this direction is the famous \kmed problem.
Here, we are given a
    set \(F\) of facilities, a set \(C\) of clients, a  metric \(d\) over \(F
    \cup C\) and an upper bound \(k\) on the number of~facilities we can open.
    A solution is a set \(S \subseteq F\) of at most \(k\) open facilities
    and a connection assignment \(\phi: C \to S\) of clients to the open facilities.
    The goal is to find a solution that minimizes the connection cost \(\sum_{c \in C}d(c, \phi(c))\).
The problem can be approximated in polynomial time up to a~constant factor~\cite{arya2004local,charikar1999constant} with the currently best approximation factor being \((2.675+\eps)\)~\cite{byrka2015improved}.
On the other hand, we cannot hope for a polynomial-time~$(1+{2}/{e}-\eps)$-approximation, since
it would entail P=NP~\cite{guha1999greedy}.
Therefore, there is a~gap in our understanding of the optimal approximability of \kmed.

Surprisingly, the situation becomes simpler if we consider parameterized algorithms, with $k$ as the~natural choice of parameterization.
Such a~parameterized problem is W[2]-hard~\cite{esa} so it is unlikely to admit an~exact algorithm with running time of the form $f(k)\cdot n^{\Oh(1)}$, where $n$ is the size of an~instance.
However, Cohen-Addad et al.~\cite{kmedian-fpt} have obtained an algorithm with approximation factor $(1+{2}/{e}+\eps)$
and running time\footnote{We omit the dependency on $\eps$ in the running time except for approximation schemes.} $2^{\Oh(k\log k)}\cdot n^{\Oh(1)}$.
This result is essentially tight, as the existence of an~FPT-approximation with factor $(1+{2}/{e}-\eps)$ would contradict the \emph{Gap Exponential Time Hypothesis}\footnote{The Gap Exponential Time Hypothesis~\cite{Dinur16,MR17} states that, for some constant $\gamma > 0$, there is no $2^{o(n)}$-time algorithm that can, given a 3SAT instance, distinguish between (1) the instance is fully satisfiable or (2) any assignment to the instance violates at least $\gamma$ fraction of the clauses.} (Gap-ETH)~\cite{kmedian-fpt}.
The mentioned hardness result has also ruled out running time of the form $f(k)\cdot n^{g(k)}$, where  $g = k^{poly(1/\eps)}$.
This lower bound has been later strengthened: under Gap-ETH no algorithm with running time $f(k)\cdot n^{o(k)}$ can achieve approximation factor $(1+{2}/{e}-\eps)$~\cite{M20}.

The parameterized approach brought also a~breakthrough to the understanding of \capkmed.
In this setting, each facility \(f\) is associated with a~capacity \(u_f \in
    \mathbb{Z}_{\geqslant 0}\) and the connection assignment $\phi$ must satisfy
    \(\size{\phi^{-1}(f)} \leqslant u_f\) for every facility \(f \in S\).
The best known polynomial-time approximation for \capkmed
is burdened with a~factor $\Oh(\log k)$~\cite{esa,DBLP:conf/stoc/CharikarCGG98} and relies on the generic technique of metric tree embeddings with expected logarithmic
    distortion~\cite{FakcharoenpholRT03}.
All
    the known constant-factor approximations violate either the number of facilities or the
    capacities.  Li  has provided  such an~algorithm  by opening \((1 +
    \eps) \cdot k\) facilities~\cite{Li15uniform,Li16non_uniform}. Afterwards
    analogous results, but violating the capacities by a~factor of \((1 +
    \eps)\) were also
    obtained~\cite{DBLP:conf/ipco/ByrkaRU16,DBLP:conf/icalp/DemirciL16}.    
This is in contrast with other
    capacitated clustering problems such as \textsc{Facility Location} or \textsc{$k$-Center},
    for which constant factor approximation algorithms have been constructed~\cite{cygan2012lp,korupolu2000analysis}.
However, no superconstant lower bound for \capkmed is known.

When it comes to parameterized algorithms, Adamczyk et al.~\cite{esa} have presented a $(7+\eps)$-approximation algorithm with running time $2^{\Oh(k\log k)}\cdot n^{\Oh(1)}$ for \capkmed.
Xu et al.~\cite{DBLP:journals/corr/abs-1901-04628} proposed a similar algorithm for the related \textsc{Capacitated} $k$-\textsc{Means} problem, where one minimizes the sum of~squares of distances. 
These results have been improved by Cohen-Addad and Li~\cite{improved}, who obtained factor $(3+\eps)$ for \textsc{Capacitated} $k$-\textsc{median} and $(9+\eps)$ for \textsc{Capacitated} $k$-\textsc{means}, within the same running time.

\subsubsection{Our contribution}
In this work, we study a different choice of parameterization for \kmed.
Whereas $k$ is the number of facilities to open, we
consider the dual parameter $\ell = |F| - k$: the number of facilities to be closed.
We refer to this problem as \colmed in order to avoid ambiguity.
Note that even though this is the same task from the perspective of polynomial-time algorithms, it is a~different problem when seen through the lens of parameterized complexity.
First, we observe that \colmed is W[1]-hard (Theorem \ref{thm:w1hard}), which motivates the study of approximation algorithms also for this choice of parameterization.
It turns out that switching to the dual parameterization changes the approximability status dramatically and we can obtain an~arbitrarily good approximation factor.
More precisely, we present an efficient parameterized approximation scheme (EPAS), i.e.,  $(1+\eps)$-approximation with running time of the form $f(\ell,\eps)\cdot n^{\Oh(1)}$.
This constitutes our main result.

\begin{theorem}
\label{thm:uncap}
The \colmed problem admits a deterministic $(1+\eps)$-approximation algorithm running in time $2^{\Oh(\ell\log(\ell/\eps))}\cdot n^{\Oh(1)}$ for any constant $\eps > 0$.
\end{theorem}

We obtain this result by combining the technique of color-coding from the FPT theory 
with a~greedy approach common in the design of approximation algorithms. 
The running time becomes polynomial whenever we want to open all but $\Oh\left(\frac{\log n}{\log\log n}\right)$ facilities.
To the best of our knowledge, this is the first non-trivial setting with general metric space which admits an approximation scheme. 

A natural question arises about the behavior of the capacitated version of~the~problem in this setting, referred to as \capcolmed.
Both in polynomial-time regime or when parameterized by $k$, there is no evidence that the capacitated problem is any harder and the gap between the approximation factors might be just a~result of our lack of understanding.
Somehow surprisingly, for the dual parameterization $\ell$ we are able to show a~clear separation between the capacitated and uncapacitated case.
Namely, we present a~reduction from the \textsc{Max $k$-Coverage} problem which entails the same approximation lower bound as for the uncapacitated problem parameterized by $k$.

\begin{theorem} \label{thm:hardness-capacitated-intro}
Assuming Gap-ETH, there is no $f(\ell) \cdot n^{o(\ell)}$-time algorithm that can approximate \capcolmed to within a factor of $(1 + 2/e - \epsilon)$ for any function~$f$ and any constant $\epsilon > 0$.
\end{theorem}

\subsubsection{Related work}
A simple example of dual parameterization is given by \textsc{$k$-Independent Set} and \textsc{$\ell$-Vertex Cover}.
From the perspective of polynomial-time algorithms, these problems are equivalent (by setting $\ell = |V(G)| - k$), but they differ greatly when analyzed as parameterized problems: the first one is W[1]-hard while the latter is FPT and admits a polynomial kernel~\cite{CyganFKLMPPS15}.  
Another example originates from the~early work on \textsc{$k$-Dominating Set}, which is a~basic W[2]-complete problem.
When parameterized by $\ell = |V(G)|-k$, the problem is known as \textsc{$\ell$-Nonblocker}.
This name can be interpreted as a~task of choosing $\ell$ vertices so that none is blocked by the others, i.e., each chosen vertex has a~neighbor which has not been chosen.
Under this parameterization, the problem is FPT and admits a~linear kernel~\cite{nonblocker1}.
The best known running time for \textsc{$\ell$-Nonblocker} is $1.96^\ell \cdot n^{\Oh(1)}$~\cite{nonblocker2}.
It is worth noting that \textsc{$\ell$-Nonblocker} is a~special case of $\colmed$ with a graph metric and $F = C = V(G)$, however this analogy works only in a~non-approximate setting.

The Gap Exponential Time Hypothesis was employed for proving parameterized inapproximability by Chalermsook et al.~\cite{from-gapeth},
who presented hardness results for \textsc{$k$-Clique}, \textsc{$k$-Dominating Set}, and \textsc{Densest $k$-Subgraph}.
It was later used to obtain lower bounds for
\textsc{Directed Odd Cycle Transversal}~\cite{doct}, \textsc{Directed Steiner Network}~\cite{chitnis2018parameterized}, \textsc{Planar Steiner Orientation}~\cite{planar-steiner-orient}, and \textsc{Unique Set Cover}~\cite{M20}, among others.
Moreover, Gap-ETH turned out to be a~sufficient assumption to~rule out the existence of an~FPT algorithm for $k$-\textsc{Even Set}~\cite{even-set}.

\section{Preliminaries}
\label{sec:prelims}

\subsubsection{Parameterized complexity and reductions}
A parameterized problem instance is
    created by associating an input instance with an integer parameter \(k\). We
    say that a problem is \emph{fixed parameter tractable} (\FPT{}) if
    it admits an~algorithm solving an~instance \((I, k)\) in time
    \(f(k)\cdot |I|^{\Oh(1)}\), where \(f\) is some computable function.
    Such an algorithm we shall call an \emph{FPT algorithm}.

    To show that a
    problem is unlikely to be \FPT{}, we use \emph{parameterized reductions} analogous
    to those employed in the classical complexity theory (see \cite{CyganFKLMPPS15}). Here, the concept of
    \textsf{W}-hardness replaces the one of \NP-hardness, and we need not only to
    construct an equivalent instance in time \(f(k)\cdot |I|^{\Oh(1)}\), but also to ensure that the
    value of the parameter in the new instance depends only on the value of the
    parameter in the original instance.
In contrast to the \NP-hardness theory, there is a hierarchy of classes $\FPT = \textsf{W[0]} \subseteq \textsf{W[1]} \subseteq \textsf{W[2]} \subseteq \dots$ and these containments are believed to be strict.
If there exists a parameterized reduction transforming a~\textsf{W[t]}-hard problem 
    to another problem \(\Pi \), then the problem \(\Pi \) is
    \textsf{W[t]}-hard as well.
    If a parameterized reduction transforms parameter linearly, i.e., maps an instance \((I_1, k)\) to \((I_2, \Oh(k))\),
    then it also preserves running time of the form $f(k)\cdot |I|^{o(k)}$.
    
In order to prove hardness of parameterized approximation, we use parameterized reductions between \emph{promise problems}.
Suppose we are given an instance \((I_1, k_1)\) of a~minimization problem
with a promise that the answer is at most $D_1$ and we want to find a~solution of value at most $\alpha \cdot D_1$.
Then a reduction should map \((I_1, k_1)\) to such an instance \((I_2, k_2)\) so that the answer to it is at at most $D_2$ and any solution to \((I_2, k_2)\) of value at most $\alpha \cdot D_2$ can be transformed in time  $f(k_1)\cdot |I_1|^{\Oh(1)}$ to a solution to \((I_1, k_1)\) of value at most $\alpha \cdot D_1$.
If an FPT $\alpha$-approximation exists for the latter problem,
then it exists also for the first one.
Again, if we have $k_2 = \Oh(k_1)$, then this relation holds also for algorithms with running time of the form $f(k)\cdot |I|^{o(k)}$.

\subsubsection{Problem definitions} Below we formally introduce the main studied problem and the problems employed in reductions.

\defparproblem{{(Capacitated) co-$\ell$-Median}}
{set of facilities \(F\), set of clients \(C\), metric \(d\) over \(F \cup C\), sequence of capacities \(u_f \in
    \mathbb{Z}_{\geqslant 0}\), integer $\ell$}{$\ell$}
{find a set \(S \subseteq F\) of at most \(|F|-\ell\) facilities
    and a connection assignment \(\phi: C \to F \setminus S\) that satisfies \(\size{\phi^{-1}(f)} \leqslant u_f\) for all $f \in F \setminus S$, and minimizes \(\sum_{c \in C}d(c, \phi(c))\)}

A metric $d: (F \cup C) \times (F \cup C) \to \mathbb{R}_{\geqslant 0}$ is a symmetric function that obeys the triangle inequality $d(x,y) + d(y,z) \ge d(x,z)$ and satisfies $d(x,x) = 0$.
In~the~uncapacitated version we assume that all capacities are equal $|C|$, so any assignment \(\phi: C \to F \setminus S\) is valid.
In the approximate version of \capcolmed we treat the capacity condition \(\size{\phi^{-1}(f)} \leqslant u_f\) as a~hard constraint and we allow only the connection cost \(\sum_{c \in C}d(c, \phi(c))\) to be larger than the optimum.

\defparproblem{{$k$-Independent Set}}{graph $G=(V,E)$, integer $k$}{$k$}
{decide whether there exists a set $S \subseteq V(G)$ of size $k$ such that for all pairs $u,v\in S$ we have $uv \not\in E(G)$}

\defparproblem{{Max $k$-Coverage}}{universe $U$, family of subsets $T_1, \dots, T_n \subseteq U$, integer $k$}{$k$}
{find $k$ subsets $T_{i_1}, \dots, T_{i_k}$ that maximizes $|T_{i_1} \cup \cdots \cup T_{i_k}|$}

\section{Uncapacitated \colmed}

We begin with a simple reduction, showing that the exact problem remains hard under the dual parameterization.

\begin{theorem}\label{thm:w1hard}
The \textsc{co-$\ell$-Median} problem is W[1]-hard.
\end{theorem}
\begin{proof}
We reduce from \textsc{$\ell$-Independent Set}, which is W[1]-hard.
We transform a given graph $G$ into a \textsc{co-$\ell$-Median} instance by setting $F = V(G)$,
and placing a client in the middle of each edge.
The distance from a client to both endpoints of its edge is 1 and the shortest paths of such subdivided graph induce the metric~$d$.

If we did not close any facilities, the cost of serving all clients would equal $|E(G)|$.
The same holds if each client has an open facility within distance 1, so the set of closed facilities forms an independent set of vertices in $G$.
On the other hand, if we close a set of facilities containing two endpoints of a single edge then the cost increases.
Therefore the answer to the created instance is $|E(G)|$ if and only if $G$ contains an independent set of size~$\ell$.
\end{proof}

We move on to designing a parameterized approximation scheme for \colmed.
We use notation $d(c,S)$ for the minimum distance between $c$ and any element of the set $S$.
In the uncapacitated setting the connection assignment $\phi_S$ is unique for a given set of closed facilities $S$: each client is assigned to the closest facility outside $S$.
Whenever we consider a solution set $S \subseteq F$,
we mean that this is the set of closed facilities
and denote $cost(S) = \sum_{c \in C} d(c, F \setminus S)$.
We
define $V(f)$ to be the \emph{Voronoi cell} of facility $f$, i.e., the set of clients for which $f$ is the closest facility.
We can break ties arbitrarily and for the sake of disambiguation we assume an~ordering on $F$ and whenever two distances are equal
we choose the facility that comes first in the ordering.

Let $C(f)$ denote the cost of the cell $V(f)$, i.e., $\sum_{c \in V(f)} d(c, f)$.
For a solution $S$ and $f\in S$, $g \not\in S$, we define $C(S,f,g) = \sum d(c,g)$ over $\{c \in V(f)\,|\, \phi_S(c) = g\}$, that is,
the sum of connections of clients that switched from $f$ to $g$.
Note that as long as $f \in F$ remains open, there is no need to change connections of the clients in $V(f)$.
We can express the difference of connection costs after closing $S$ as

$$\Delta(S) = \sum_{f \in S}\sum_{c \in V(f)} d(c, F \setminus S) - \sum_{f \in S} C(f) = \sum_{f \in S}\sum_{g \in F \setminus S} C(S,f,g) - \sum_{f \in S} C(f).$$

We have $cost(S) = \sum_{f\in F} C(f) + \Delta(S)$, therefore the optimal solution closes set $S$ of size $\ell$ minimizing $\Delta(S)$.

The crucial observation is that any~small set of closed facilities $S$ can be associated with a~small set of open facilities that are relevant for serving the clients from $\bigcup_{f\in S} V(f)$.
Intuitively, if $C(S, f,g) = \Oh(\frac{\eps}{\ell^2}) \cdot cost(S)$ for all $f\in S$, then we can afford replacing $\ell$ such facilities $g$ with others that are not too far away.

\begin{definition}\label{def:support}
The $\eps$-support of a~solution $S \subseteq F$, $|S|=\ell$, referred to as $\support(S)$, is the set of all open facilities $g$ (i.e., $g \not\in S$) satisfying one of the following conditions:
\begin{enumerate}
\item there is $f \in S$ such that $g$ minimizes distance $d(f,g)$ among all open facilities,
    \item there is $f \in S$ such that $C(S, f,g) > \frac{\eps}{6\ell^2} \cdot cost(S)$.
\end{enumerate}

We break ties in condition (1) according to the same rule as in the definition of $V(f)$, so there is a single $g$ satisfying condition (1) for each~$f$.
\end{definition}

\begin{lemma}\label{lem:support-size}
For a solution $S$ of size $\ell$, we have $|\support(S)| \le 6\cdot \ell^3/\eps + \ell$.
\end{lemma}
\begin{proof}
We get at most $\ell$ facilities from condition (1).
Since the sets of clients being served by different $g \in F \setminus S$ are disjoint and $\sum_{g \in F} C(S, f,g) \le cost(S)$,
we obtain at most $6\cdot \ell^3/\eps$ facilities from condition (2).
\end{proof}

Even though we will not compute the set $\support(Opt)$ directly, we are going to work with partitions $F = A\, \uplus\, B$, such that $Opt \subseteq A$ and $\support(Opt) \subseteq B$.
Such a partition already gives us a~valuable hint.
By looking at each facility $f \in A$ separately, we can deduce that if $f \in Opt$ and some other facility $g$ belongs to $A$ (so it cannot belong to $\support(Opt)$) then in some cases $g$ must also belong to $Opt$.
More precisely, if $g \in A$ is closer to $f$ than the closest facility in $B$, then $g$ must be closed, as otherwise it would violate condition (1).
Furthermore, suppose that $h \in A$ serves clients from $V(f)$ (assuming $f$ is closed) of total cost at least $\frac{\eps}{6\ell^2} \cdot cost(S)$.
If we keep $h$ open and close some other facilities, this relation is preserved and having $h$ in $A$ violates condition (2).
We formalize this idea with the notion of \emph{required sets},
given by the following procedure, supplied additionally with a~real number~$D$, which can be regarded as the guessed value of $cost(Opt)$.

\begin{algorithm}
\caption{\textsc{Compute-required-set$(A,B,f,\eps,\ell,D)$} (assume $f \in A$ and $A \cap B =\emptyset$)}
\label{alg:of}
\begin{algorithmic}[1]
\STATE $s_f \leftarrow \min_{y \in B} d(f,y)$
\STATE $R_f \leftarrow \{g \in A\, \mid \, d(f,g) < s_f\}$ \text{ (including} $f$)
\WHILE {$\exists {g \in A}\, :\, C(R_f, f, g) > \frac{\eps}{3\ell^2} \cdot D$}
\STATE $R_f \leftarrow R_f \cup \{g\}$
\ENDWHILE
\RETURN $R_f$
\end{algorithmic}
\end{algorithm}

\begin{lemma}\label{lem:required-sets}
Let $Opt \subseteq F$ be the optimal solution.
Suppose $F = A\, \uplus\, B$, $f\in Opt \subseteq A$, $\support(Opt) \subseteq B$, and $cost(Opt) \le 2D$.
Then the set $R_f$ returned by \textsc{Compute-required-set$(A,B,f,\eps,\ell,D)$} satisfies $R_f \subseteq Opt$.
\end{lemma}
\begin{proof}
Let $y_f$ be the facility in $B$ that is closest to $f$.
Due to condition (1) in Definition~\ref{def:support}, all facilities $g \in A$ satisfying $d(f,g) < d(f,y_f)$ must be closed in the optimal solution, so we initially add them to~$R_f$.
We keep invariant $R_f \subseteq Opt$, so for any $g \in F \setminus Opt$ it holds that $C(Opt, f, g) \ge C(R_f, f, g)$.
Whenever there is $g \in A$ satisfying $C(R_f, f, g) > \frac{\eps}{3\ell^2} \cdot D$, we get
$$C(Opt, f, g) \ge C(R_f, f, g) > \frac{\eps}{3\ell^2} \cdot D \ge \frac{\eps}{6\ell^2} \cdot cost(Opt).$$
Since $g$ does not belong to $\support(Opt) \subseteq B$,
then by condition (2) it must be closed.
Hence, adding $g$ to $R_f$ preserves the invariant.
\end{proof}

Before proving the main technical lemma,
we need one more simple observation, in which we exploit the fact that the function $d$ is indeed a~metric.

\begin{lemma}\label{lem:reroute}
Suppose $c \in V(f_0)$ and $d(f_0,f_1) \le d(f_0,f_2)$.
Then $d(c,f_1) \le 3\cdot d(c,f_2)$.
\end{lemma}
\begin{proof}
An illustration is given in Figure~\ref{fig:voronoi}.
Since $c$ belongs to the Voronoi cell of~$f_0$, we~have $d(c,f_0) \le d(c,f_2)$.
By the triangle inequality 
\begin{align*}
d(c,f_1) \le & d(c,f_0) + d(f_0,f_1) \le d(c,f_0) + d(f_0,f_2) \le \\
& d(c,f_0) + d(c,f_0) + d(c,f_2) \le 3\cdot d(c,f_2).
\end{align*}
\end{proof}

\begin{figure}
    \centering
    \includegraphics[scale=0.6]{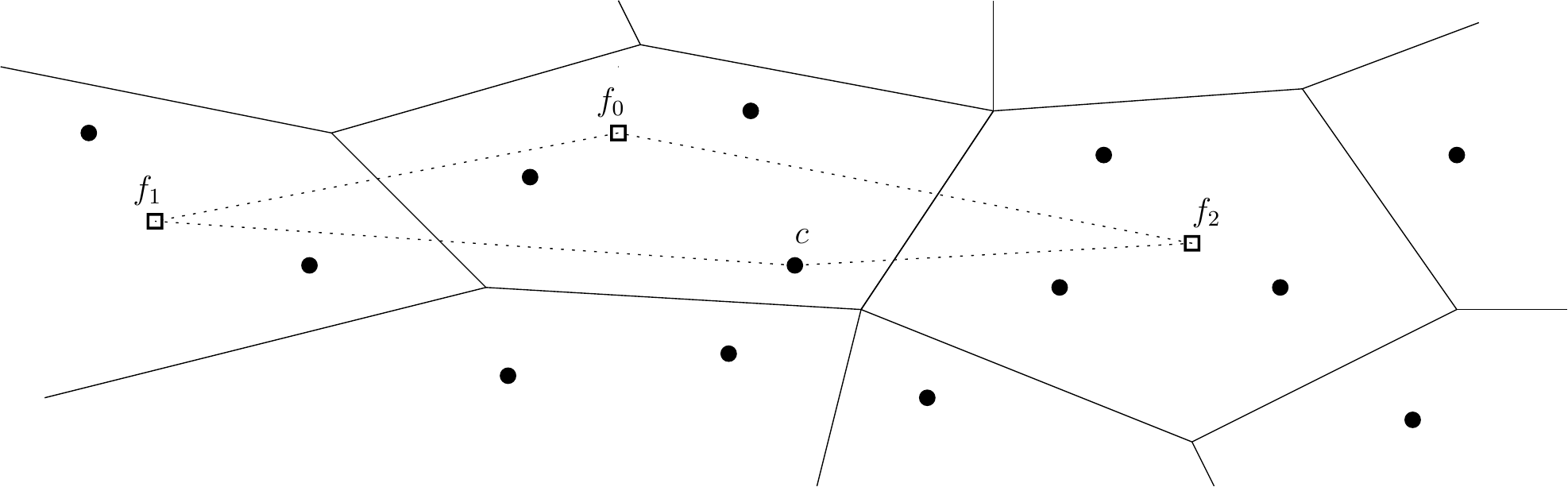}
    \caption{An example of a Voronoi diagram with squares representing facilities and dots being clients.
    Lemma \ref{lem:reroute} states that even if $d(c,f_1) > d(c,f_2)$ for $c \in V(f_0)$ and $d(f_0,f_1) \le d(f_0,f_2)$, then $d(c,f_1)$ cannot be larger than $3\cdot d(c,f_2)$. 
    }
    \label{fig:voronoi}
\end{figure}

\begin{lemma}\label{lem:given-partition}
Suppose we are given a partition $F = A\, \uplus\, B$, such that $Opt \subseteq A$, $\support(Opt) \subseteq B$, and a number $D \in \mathbb{R}_{> 0}$, such that $cost(Opt) \in [D,2D]$.
Then we can find a solution $S \subseteq A$, such that $cost(S) \le (1+\eps)\cdot cost(Opt)$, in polynomial time.
\end{lemma}
\begin{proof}
We compute the set $R_f = \textsc{Compute-required-set}(A,B,f,\eps,\ell,D)$
for each facility $f \in A$.
The subroutine from Algorithm~\ref{alg:of} clearly runs in polynomial time.
Furthermore, for each $f \in A$ we
compute its \emph{marginal cost} of closing
$$m_f = \sum_{c \in V(f)} d(c, F \setminus R_f) - C(f).$$

If $|R_f| > \ell$ then $f$ cannot belong to any solution consistent with the partition $(A,B)$ and in this case we set $m_f = \infty$.
Since the marginal cost depends only on $f$, we can greedily choose $\ell$ facilities from $A$ that minimize $m_f$ -- we refer to this set as $S$.

We first argue that $\sum_{f \in F} C(f) + \sum_{f \in S} m_f$ is at most the cost of the optimal solution.
By greedy choice we have that
$\sum_{f \in S} m_f \le \sum_{f \in Opt} m_f$.
We have assumed $cost(Opt) \le 2D$ so by Lemma~\ref{lem:required-sets} we get that if $f \in Opt$, then $R_f \subseteq Opt$.
The set of facilities $F \setminus Opt$ that can serve clients from $V(f)$ is a~subset of~$F \setminus R_f$ and the distances can only increase,
thus for $f \in Opt$ we have $ m_f \le \sum_{c \in V(f)} d(c, F \setminus Opt) - C(f)$.
We conclude that $\sum_{f \in F} C(f) + \sum_{f \in S} m_f$ is upper bounded by
\begin{equation}\label{eq:marginal}
\sum_{f \in F} C(f) + \sum_{f \in Opt} \sum_{c \in V(f)} d(c, F \setminus Opt) - \sum_{f \in Opt} C(f) = cost(Opt).
\end{equation}

The second argument is that after switching benchmark from the marginal cost to the true cost of closing $S$, we will additionally pay at most $\eps D$.
These quantities differ when for a facility $f\in S$ we have `connected' some clients from $V(f)$ to $g \in S \setminus R_f$ when computing $m_f$.
More precisely, we want to show that for each $f\in S$ we have

\begin{equation}\label{eq:eps-d}
\sum_{c \in V(f)} d(c, F \setminus S) \le \sum_{c \in V(f)} d(c, F \setminus R_f) + \frac{\eps D}{\ell}.
\end{equation}

By the construction of $R_f$, whenever $g \in S \setminus R_f$ we are guaranteed that there exists a~facility $y \in B$ such that $d(f,g) \ge d(f,y)$ and, moreover, $C(R_f, f, g) \le \frac{\eps}{3\ell^2} \cdot D$.
We can reroute all such clients $c$ to the closest open facility and we know it is not further than $d(c,y)$.
By Lemma~\ref{lem:reroute} we know that $d(c,y) \le 3\cdot d(c,g)$ so rerouting those clients costs at most $\frac{\eps}{\ell^2} \cdot D$.
Since there are at most $\ell$ such facilities  $g \in S \setminus R_f$, we have proved Formula~(\ref{eq:eps-d}).
Combining this with bound from~(\ref{eq:marginal}) implies that $cost(S) \le cost(Opt) + \eps D$.
As we have assumed $D \le cost(Opt)$, the~claim follows.
\end{proof}

In order to apply Lemma~\ref{lem:given-partition}, we need to find
a partition $F = A \uplus B$ satisfying $Opt \subseteq A$ and $\support(Opt) \subseteq B$.
Since $\support(Opt) = \Oh(\ell^3 / \eps)$,
we can do this via randomization.
Consider tossing a biased coin for each facility independently: with probability $\frac{\eps}{\ell^3}$ we place it in $A$, and with remaining probability in $B$.
The probability of obtaining a partitioning satisfying $Opt \subseteq A$ and $\support(Opt) \subseteq S$ equals $(\frac{\eps}{\ell^3})^\ell$ times $(1-\frac{\eps}{\ell^3})^{O(\frac{\ell^3}{\eps})} = \Omega(1)$.
Therefore $2^{O(\ell\log(\ell/\eps))}$ trials give a~constant probability of sampling a correct partitioning.
In order to derandomize this process, we take advantage of the following construction which is a~folklore corollary from
the framework of {color-coding}~\cite{color-coding}.
As we are not aware of any self-contained proof of this claim in the literature, we provide it for completeness.

\begin{lemma}
\label{lem:splitter}
For a set $U$ of size $n$,
there exists a family $\mathcal{H}$ of partitions $U = A\, \uplus\, B$ such that $|\mathcal{H}| = 2^{O(\ell\log(\ell+r))}\log n$ and for every pair of disjoint sets $A_0,B_0 \subseteq U$ with $|A_0| \le \ell$, $|B_0| \le r$, there is $(A,B) \in \mathcal{F}$ satisfying $A_0 \subseteq A, B_0 \subseteq B$.
The~family $\mathcal{H}$ can be constructed in time $2^{O(\ell\log(\ell+r))}n\log n$.
\end{lemma}
\begin{proof}
Let use denote $[n] = \{1, 2, \dots, n\}$ and identify $U = [n]$.
We rely on the following theorem: for any integers $n,k$
there exists a family $\mathcal{F}$ of functions $f: [n] \rightarrow [k^2]$, such that  $|\mathcal{F}| = k^{\Oh(1)}\log n$ and for each $X \subseteq [n]$ of size $k$ there is a function $f \in  \mathcal{F}$ which is injective on $X$; moreover,  $\mathcal{F}$ can be constructed in~time $k^{\Oh(1)}n\log n$~\cite[Theorem 5.16]{CyganFKLMPPS15}.

We use this construction for $k= \ell + r$.
Next, consider the family $\mathcal{G}$ of all functions $g: [(\ell+r)^2] \rightarrow \{0,1\}$ such that $|g^{-1}(0)| \le \ell$.
Clearly, $|\mathcal{G}| \le (\ell + r)^{2\ell}$.
The family $\mathcal{H}$ is given by taking all compositions $\{h = g \circ f\, \mid \, g \in \mathcal{G}, f \in \mathcal{F}\}$
and setting $(A_h,B_h) = (h^{-1}(0), h^{-1}(1))$.
We have $|\mathcal{H}| \le |\mathcal{G}| \cdot |\mathcal{F}| = 2^{O(\ell\log(\ell+r))}\log n$.
Let us consider any pair of disjoint subsets $A_0,B_0 \subseteq [n]$ with $|A_0| \le \ell$, $|B_0| \le r$.
There exists $f \in  \mathcal{F}$ injective on $A_0 \cup B_0$
and $g \in  \mathcal{G}$ that maps $f(A_0)$ to 0 and $f(B_0)$ to 1, so
 $A_0 \subseteq A_{g \circ f}, B_0 \subseteq B_{g \circ f}$.
\end{proof}

\begin{theorem}
The \colmed problem admits a deterministic $(1+\eps)$-approximation algorithm running in time $2^{\Oh(\ell\log(\ell/\eps))}\cdot n^{\Oh(1)}$ for any constant $\eps > 0$.
\end{theorem}
\begin{proof}
We apply Lemma~\ref{lem:splitter} for $U=F$, $\ell$ being the parameter, and $r =  6\cdot \ell^3/\eps + \ell$, which upper bounds the size of $\support(Opt)$ (Lemma~\ref{lem:support-size}).
The family $\mathcal{H}$ contains a~partition $F = A \uplus B$ satisfying $Opt \subseteq A$ and $\support(Opt) \subseteq B$.
Next, we~need to find $D$, such that $cost(Opt) \in [D,2D]$.
We begin with any polynomial-time $\alpha$-approximation algorithm for \kmed ($\alpha = \Oh(1)$) to get an interval $[X, \alpha X]$, which contains $cost(Opt)$.
We cover this interval with a~constant number of intervals of the form $[X, 2X]$ and one of these provides a~valid value of $D$.
We~invoke the algorithm from Lemma~\ref{lem:given-partition} for each such triple $(A,B,D)$ and return a~solution with the smallest cost.
\end{proof}

\section{Hardness of \capcolmed}

In this section we show that, unlike \textsc{co-$\ell$-Median}, its capacitated counterpart does not admit a parameterized approximation scheme.

We shall reduce from the \textsc{Max $k$-Coverage} problem,
which was also the source of lower bounds for \kmed in the polynomial-time regime~\cite{guha1999greedy} and when parameterized by $k$~\cite{kmedian-fpt}.
However, the latter reduction is not longer valid when we consider a~different parameterization for~\kmed, as otherwise we could not obtain Theorem~\ref{thm:uncap}.
Therefore, we need to design a~new reduction, that exploits the capacity constraints and translates the parameter $k$ of~an~instance of~\textsc{Max $k$-Coverage} into the parameter $\ell$ of
an~instance of \capcolmed.
To the best of our knowledge, this is the first hardness result in which the capacities play a role and allow us to obtain a~better lower bound.

We rely on the following strong hardness result.
Note that this result is a~strengthening of~\cite{kmedian-fpt}, which only rules out $f(k) \cdot n^{k^{poly(1/\delta)}}$-time algorithm. This suffices to rule out a~parameterized approximation scheme for \capcolmed, but not for a strong running time lower bound of the form $f(\ell) \cdot n^{o(\ell)}$.

\begin{theorem}[\cite{M20}] \label{thm:max-coverage}
Assuming Gap-ETH, there is no $f(k) \cdot n^{o(k)}$-time algorithm that can approximate \textsc{Max $k$-Coverage} to within a factor of $(1 - 1/e + \delta)$ for any function $f$ and any constant $\delta > 0$. Furthermore, this holds even when every input subset is of the same size and with a promise that there exists $k$ subsets that covers each element exactly once.
\end{theorem}

We can now prove our hardness result for \capcolmed.

\begin{theorem} \label{thm:hardness-capacitated}
Assuming Gap-ETH, there is no $f(\ell) \cdot n^{o(\ell)}$-time algorithm that can approximate \capcolmed to within a factor of $(1 + 2/e - \epsilon)$ for any function $f$ and any constant $\epsilon > 0$.
\end{theorem}
\begin{proof}
Let $U, T_1, \dots, T_n$ be an instance of \textsc{Max $k$-Coverage}. We create an~instance $(F, C)$ of \capcolmed as follows.
\begin{itemize}
\item For each subset $T_i$ with $i \in [n]$, create a facility $f^{set}_i$ with capacity $|T_i|$. For each element $u \in U$, create a facility $f^{element}_u$ with capacity $|U| + 2$.
\item For every $i \in [n]$, create $|T_i|$ clients $c_{i, 1}^{set}, \dots, c_{i, |T_i|}^{set}$. For each $j \in [|T_i|]$, we~define the distance from $c_{i, j}^{set}$ to the facilities by
\begin{align*}
d(c_{i, j}^{set}, f^{set}_i) &= 0, \\
d(c_{i, j}^{set}, f^{element}_u) &= 1 &\forall u \in T_i, \\
d(c_{i, j}^{set}, f^{set}_{i'}) &= 2 &\forall i' \ne i, \\
d(c_{i, j}^{set}, f^{element}_u) &= 3 &\forall u \notin T_i.
\end{align*}
\item For every element $u \in U$, create $|U| + 1$ clients $c_{u, 1}^{element}, \dots, c_{u, |U| + 1}^{element}$ and, for each $j \in [|U| + 1]$, define the distance from $c_{u, j}^{element}$ to the facilities by
\begin{align*}
d(c_{u, j}^{element}, f_u^{element}) &= 0, \\
d(c_{u, j}^{element}, f_i^{set}) &= 1 &\forall T_i \ni u, \\
d(c_{u, j}^{element}, f_{u'}^{element}) &= 2 &\forall u' \ne u, \\
d(c_{u, j}^{element}, f_i^{set}) &= 3 &\forall T_i \not\ni u.
\end{align*}
\item Let $\ell = k$.
\end{itemize}

Suppose that we have an $f(\ell) \cdot n^{o(\ell)}$-time  $(1 + 2/e - \epsilon)$-approximation algorithm for \capcolmed. We will use it to approximate \textsc{Max $k$-Coverage} instance with $|T_1| = \cdots = |T_n| = |U|/k$ with a promise that there exists $k$ subsets that covers each element exactly once, as follows. We run the above reduction to produce an instance $(F, C)$ and run the approximation algorithm for \capcolmed; let $S \subseteq F$ be the produced solution. Notice that $S$ may not contain any element-facility, as otherwise there would not even be enough capacity left to serve all clients. Hence, $S = \{f^{set}_{i_1}, \dots, f^{set}_{i_k}\}$. We claim that $T_{i_1}, \dots, T_{i_k}$ is an $(1 - 1/e + \epsilon/2)$-approximate solution for \textsc{Max $k$-Coverage}.

To see that $T_{i_1}, \dots, T_{i_k}$ is an $(1 - 1/e + \epsilon/2)$-approximate solution for \textsc{Max $k$-Coverage}, notice that the cost of closing $\{f^{set}_{i_1}, \dots, f^{set}_{i_k}\}$ is exactly $|T_{i_1} \cup \cdots \cup T_{i_k}| + 3 \cdot |U \setminus (T_{i_1} \cup \cdots \cup T_{i_k})|$ because each element-facility $f_u^{element}$ can only serve one more client in addition to $c_{u, 1}^{element}, \dots, c_{u, |U| + 1}^{element}$. (Note that we may assume without loss of generality that $f_u^{element}$ serves $c_{u, 1}^{element}, \dots, c_{u, |U| + 1}^{element}$.) Moreover, there are exactly $|T_{i_1}| + \cdots + |T_{i_k}| = |U|$ clients left to be served after the closure of $\{f^{set}_{i_1}, \dots, f^{set}_{i_k}\}$. Hence, each element-facility $f_u^{element}$ with $u \in T_{i_1} \cup \cdots \cup T_{i_k}$ can serve a client of distance one from it. All other element-facilities will have to serve a client of distance three from it. This results in the cost of exactly $|T_{i_1} \cup \cdots \cup T_{i_k}| + 3 \cdot |U \setminus (T_{i_1} \cup \cdots \cup T_{i_k})|$. Now, since we are promised that there exists $k$ subsets that uniquely covers the universe $U$, the optimum of \capcolmed must be $|U|$. Since our (assumed) approximation algorithm for \capcolmed has approximation factor $(1 + 2/e - \epsilon)$, we must have $|T_{i_1} \cup \cdots \cup T_{i_k}| + 3 \cdot |U \setminus (T_{i_1} \cup \cdots \cup T_{i_k})| \leq |U| \cdot (1 + 2/e - \epsilon)$, which implies that $|T_{i_1} \cup \cdots \cup T_{i_k}| \geq |U| \cdot (1 - 1/e + \epsilon/2)$. Hence, the proposed algorithm is an~$f(k) \cdot n^{o(k)}$-time algorithm that approximates \textsc{Max $k$-Coverage} to within a~factor of $(1 - 1/e + \epsilon/2)$, which by Theorem~\ref{thm:max-coverage} contradicts Gap-ETH.
\end{proof}

\section{Conclusions and open problems}
We have presented a parameterized approximation scheme for \colmed and shown that its capacitated version does not admit such a~scheme. 
It remains open whether \capcolmed admits any constant-factor FPT approximation.
Obtaining such a result might be an~important step towards getting a~constant-factor polynomial-time approximation,
which is a~major open problem.

Another interesting question concerns whether one can employ the framework of lossy kernelization~\cite{lossy} to get a~polynomial size approximate kernelization scheme (PSAKS) for \colmed, which would be a~strengthening of our main result.
In other words, can we process an~instance $\mathcal{I}$ in polynomial time to produce an~equivalent instance $\mathcal{I'}$ of size $poly(\ell)$ so that solving $\mathcal{I'}$ would provide a~$(1+\eps)$-approximation for $\mathcal{I}$?

\newpage

\bibliographystyle{splncs04}
\bibliography{refs}

\end{document}